\documentclass[11pt]{article}

\usepackage{bm}
\usepackage{epsf,amsfonts,amsmath}
\usepackage{amsmath}
\usepackage{amssymb}
\usepackage{amsthm}
\newtheorem{thm}{Theorem}[subsection]
\usepackage{hyperref}
\usepackage[dvips]{epsfig,graphics}
\usepackage{graphicx}
\usepackage{texdraw}
\usepackage{authblk}
\usepackage{url}

\newcommand{\field}[1]{\mathbb{#1}}
\newtheorem{theorem}[thm]{Theorem}
\newtheorem{proposition}[thm]{Proposition}
\newtheorem{lemma}[thm]{Lemma}

\newtheorem{remark}[thm]{Remark}
\DeclareMathOperator{\tr}{tr}

\DeclareMathOperator{\str}{str}
\DeclareMathOperator{\sdet}{sdet}

\title{\textbf{A noncommutative discrete potential KdV lift}}
\author[1,2]{S. Konstantinou-Rizos\thanks{skonstantin84@gmail.com, skonstantin@chesu.ru}}
\author[3]{T.  E. Kouloukas\thanks{T.E.Kouloukas@kent.ac.uk}}
\affil[1]{Institute of Mathematical Physics and Seismodynamics, Chechen State 
University, Russia}
\affil[2]{Faculty of Mathematics and Computer Technology, Chechen State 
University, Russia}
\affil[3]{School of Mathematics, Statistics \& Actuarial Science, University of Kent, UK}

\begin{document}

\maketitle

\begin{abstract}
In this paper, we construct a Grassmann extension of a Yang-Baxter map which first appeared in \cite{kouloukas} and can be considered as a lift of the discrete potential Korteweg-de Vries (dpKdV) 
equation. This noncommutative extension satisfies the Yang-Baxter equation, and it admits a $3 \times 3$ Lax matrix. Moreover, we show that it can be squeezed down to a system of lattice equations which possesses a Lax representation and whose bosonic limit is the dpKdV equation. Finally, we consider commutative analogues of the constructed Yang-Baxter map and its associated quad-graph system, and we discuss their integrability.
\end{abstract}

\hspace{.2cm} PACS numbers: 02.30.Ik 

\hspace{.2cm} Mathematics Subject Classification: 15A75, 35Q53, 39A14, 81R12.

\hspace{.2cm} Keywords: Noncommutative dpKdV, Grassmann extensions of Yang-Baxter maps, 

\hspace{.2cm} Grassmann algebras, Grassmann extensions of discrete integrable systems.

\section{Introduction}
Noncommutative extensions of integrable systems have been of extensive interest over the last few decades. Recently, in \cite{Georgi-Sasha}, a method was presented for constructing super (Grasmmann-extended) differential-difference and difference-difference systems via noncommutative extensions of Darboux transformations. This motivated further study of discrete integrable systems on Grassmann algebras \cite{Xue-Levi-Liu, Xue-Liu, Xue-Liu-2}, as well as Grassmann extensions of Yang-Baxter maps \cite{GKM, Sokor-Sasha-2}.

The theory of Yang-Baxter maps, namely set-theoretical solutions of the Yang-Baxter equation, has a fundamental role in the theory of integrable systems. The study of such solutions was formally proposed by Drinfel'd in \cite{Drinfel'd}. Yang-Baxter maps and several connections with discrete integrable systems were extensively studied over the past couple of decades (indicatively we refer to \cite{ABS-2005, Buchstaber, Etingof, Pavlos, Pap-Tongas-Veselov, Sklyanin, Veselov, Veselov3}). Lax representations \cite{Veselov2} of Yang-Baxter maps are of particular interest, since they are associated with refactorization problems of polynomial matrices, invariant spectral curves of transfer maps \cite{Veselov, Veselov3}, $r$-matrix Poisson structures \cite{kouloukas, Kouloukas2}, and they also possess a natural connection with integrable partial differential equations through Darboux transformations  \cite{Sokor-Sasha, miky}. Recently, a ``lifting'' method for constructing two-field equations, which can be viewed as Yang-Baxter maps, via equations on quad-graphs, was presented in \cite{Pap-Tongas}.

 The relation between the Yang-Baxter property of maps and the multi-dimensional consistency property of integrable equations on quad-graphs is well established. 
 The main approach is based on the symmetry analysis of the corresponding equations \cite{Pap-Tongas0,Pap-Tongas-Veselov}.
 From another point of view, under some conditions 3D consistent equations can be lifted to four dimensional Yang-Baxter maps. This lifting procedure has been described in \cite{Pap-Tongas}.  
Conversely, Yang-Baxter maps of a specific form can be squeezed down to lattice equations and systems that admit a Lax representation derived from the Lax matrix of the corresponding Yang-Baxter map. Motivated by these results and working in the direction of extending the theory of Yang-Baxter maps to the noncommutative case, in this paper, we construct a Grassmann extension of a Yang-Baxter map, first derived in \cite{kouloukas}, and we show that it possesses the Yang-Baxter property. Furthermore, we use this map to construct a noncommutative extension of the well-celebrated dpKdV equation \cite{Hirota-KdV, Frank, PNC}, also referred to as $H_1$ equation in the ABS list \cite{ABS-2004}. Both the Yang-Baxter map and its associated quad-graph system, which we construct, admit Lax representations with $3\times 3$ Lax matrices and their bosonic limit leads to the original discrete systems.

Finally, we consider all the variables of the constructed $3\times 3$ Lax matrices to be commutative, and we derive a higher-dimensional commutative extension of the original Yang-Baxter  map. This extension constitutes a symplectic map with respect to the Sklyanin bracket, and completely integrable in the Liouville sense.  As in the noncommutative case, this map can be squeezed down to a corresponding  integrable quad-graph system. 

The paper is organised as follows. In the next section, we begin with some preliminaries on Grassmann algebras, Yang-Baxter maps, quadrilateral equations and Lax representations. 
In section \ref{dpKdV-lift}, we present the construction of the Grassmann extension of the lift of the dpKdV equation, and we show that it satisfies the Yang-Baxter equation. In section \ref{lift-down}, we show that the constructed Yang-Baxter map can be reduced to a discrete quad-graph system, which constitutes a noncommutative extension of the dpKdV equation.  The Lax representation of this system is derived from the corresponding Lax matrix of the Yang-Baxter map. Section \ref{commutative} deals with commutative analogues of the systems presented in sections \ref{dpKdV-lift} and \ref{lift-down}. Finally, in section \ref{conclusions}, we close with some remarks and perspectives for future work.

\section{Preliminaries}\label{Preliminaries}
\subsection{Grassmann algebra}
Here, we give all the definitions related to Grassmann algebras that are essential for this paper. However, for more information on Grassmann analysis one can consult \cite{Berezin}.

Consider $G$ to be a $\field{Z}_2$-graded algebra over $\field{C}$. Thus, $G$, as a linear space, is a direct sum $G=G_0\oplus G_1$ (mod 2), such that $G_iG_j\subseteq G_{i+j}$. The elements of $G$ that belong either to $G_0$ or to $G_1$ are called \textit{homogeneous}, the ones in $G_1$ are called \textit{odd} (or {\sl fermionic}), while those in $G_0$ are called \textit{even} (or {\sl bosonic}).

The parity $|a|$ of an even homogeneous element $a$ is $0$, while it is $1$ for odd homogeneous elements, by definition. The parity of the product $|ab|$ of two homogeneous elements is a sum of their parities: $|ab|=|a|+|b|$. Now, for any homogeneous elements $a$ and $b$, Grassmann commutativity means that $ba=(-1)^{|a||b|}ab$ . This implies that if $\alpha\in G_1$, then $\alpha^2=0$, and 
$\alpha a=a \alpha$, for any $a\in G_0$. 

The notions of the determinant and the trace of a matrix in $G$ are defined for square matrices, $M$, of the following block-form\footnote{Note that the block matrices are not necessarily square matrices.}
\begin{equation*}
M=\left(
\begin{matrix}
 P & \Pi \\
 \Lambda & L
\end{matrix}\right).
\end{equation*}
The blocks $P$ and $L$ are matrices with even entries, while $\Pi$ and $\Lambda$ possess only odd entries. In particular, the \textit{superdeterminant} of $M$, which is usually denoted by $\sdet(M)$, is defined to be the following quantity
\begin{equation*}
\sdet(M)=\det(P-\Pi L^{-1}\Lambda)\det(L^{-1})=\det(P^{-1})\det(L-\Lambda P^{-1}\Pi),
\end{equation*}
where $\det(\cdot)$ is the usual determinant of a matrix, while the \textit{supertrace}, which is usually denoted by $\str(M)$, is defined as 
\begin{equation*}
\str(M)=\tr (P)-\tr (L),
\end{equation*}
where  $\tr(\cdot)$ is the usual trace of a matrix.

Henceforth, we adopt the following notation: We denote all even variables in $G_0$ by Latin letters, whereas for odd variables in $G_1$ we use Greek letters.
\subsection{Yang-Baxter equation and Lax representations in the Grassmann case}
Let $S: V_G\times V_G \rightarrow V_G\times V_G$ be a map
\begin{equation}\label{Smap}
((x,\chi),(y,\psi))\stackrel{S}{\mapsto}\left((u,\xi),(v,\eta)\right),
\end{equation}
where $V_G=\{(a,\alpha)\,|\, a\in G_0,\  \alpha\in G_1\}$. Map \eqref{Smap} is said to be a \textit{Grassmann extended} Yang-Baxter map, if it satisfies the 
\textit{Yang-Baxter equation}:
\begin{equation*}
S^{12}\circ S^{13} \circ S^{23}=S^{23}\circ S^{13} \circ S^{12},
\end{equation*}
where the maps $S^{ij}: V_G\times V_G\times V_G \rightarrow V_G\times V_G\times V_G$, $i,j=1,2,3$, $i\neq j$, are defined by the following relations
\begin{equation*}
S^{12}=S\times id, \quad S^{23}=id\times S \quad \text{and} \quad S^{13}=\pi^{12} S^{23} \pi^{12},
\end{equation*}
where $\pi^{12}$ is the involution defined by $\pi^{12}((x,\chi),(y,\psi),(z,\zeta))=((y,\psi),(x,\chi),(z,\zeta))$. It is obvious from the above definition that definitions of Yang-Baxter maps in the Grassmann and commutative cases coincide. The only difference is the set of the objects of the maps. Additionally, map \eqref{Smap} is called \textit{reversible} if $S^{21}\circ S^{12}=id$; map $S^{21}$ is defined as $S^{21}=\pi^{12}S^{12}\pi^{12}$.

Furthermore, if two parameters $a,b\in G_0$ are involved in the definition of \eqref{Smap}, namely we have a map
\begin{equation}\label{SPmap}
S_{a,b}: ((x,\chi),(y,\psi))\mapsto\left((u,\xi),(v,\eta)\right),
\end{equation}
satisfying the \textit{parametric Yang-Baxter equation}
\begin{equation}\label{SYB_eq1}
S^{12}_{a,b}\circ S^{13}_{a,c} \circ S^{23}_{b,c}=S^{23}_{b,c}\circ S^{13}_{a,c} \circ S^{12}_{a,b},
\end{equation}
we shall be using the term \textit{Grassmann extended parametric Yang-Baxter map}.

According to \cite{Veselov2}, a \textit{Lax matrix} of the parametric Yang-Baxter map (\ref{SPmap}), is a matrix $\mathcal{L}$, with Grassmann-valued entries in this case, depending on the point $(x,\chi) \in V_G$, a parameter $a$ and a spectral parameter $\lambda$, such that 
\begin{equation}\label{eqLax}
\mathcal{L}_a(u,\xi;\lambda)\mathcal{L}_b(v,\eta;\lambda)=\mathcal{L}_b(y,\psi;\lambda)\mathcal{L}_a(x,\chi;\lambda).
\end{equation}
The refactorization equation (\ref{eqLax}) does not always admit a unique solution with respect to $u,\xi,v,\eta$. In fact, if (\ref{eqLax}) is equivalent to $\left((u,\xi),(v,\eta)\right)=S_{a,b}((x,\chi),(y,\psi))$, then the Lax matrix $\mathcal{L}$ is said to be \textit{strong} \cite{Kouloukas2}.  
Moreover, general solutions of (\ref{eqLax}) are not always Yang-Baxter maps. 
For the Yang-Baxter property one may use the following \textit{trifactorisation criterion}: If the following matrix refactorisation problem
\begin{align*}
\mathcal{L}_a(u,\xi;&\lambda)\mathcal{L}_b(v,\eta;\lambda)\mathcal{L}_c(w,\gamma;\lambda)=\mathcal{L}_a(x,\chi;\lambda)\mathcal{L}_b(y,\psi;\lambda)\mathcal{L}_c(z,\zeta;\lambda),
\intertext{implies}
&(u,\xi)=(x,\chi), \quad (v,\eta)=(y,\psi)\quad \text{and} ~~ (w,\gamma)=(z,\zeta),
\end{align*}
then map \eqref{SPmap} is a Yang-Baxter map \cite{kouloukas, Veselov}.

Similarly to the commutative case, the quantity $\str(\mathcal{L}_b(y,\psi;\lambda)\mathcal{L}_a(x,\chi;\lambda))$ constitutes a generating function of invariants for map $S_{a,b}${\footnote{Similarly, the invariants of the transfer maps \cite{Veselov} are derived from the supertrace of the corresponding monodromy matrix.}. This can be verified by applying the supertrace to both parts of the Lax equation \eqref{eqLax}.

\subsection{Integrable equations on quad-graphs} 
Let $w$ be a function of two discrete variables $n$ and $m$. Let also $\mathcal{S}$ and $\mathcal{T}$ be the shift operators in the $n$ and $m$ direction of a two-dimensional lattice, respectively. We shall be using the notation: $w_{00}\equiv w$, $w_{ij}=S^iT^j f$; for example, $w_{10}=w(n+1,m)$, $w_{01}=w(n,m+1)$ and $w_{11}=w(n+1,m+1)$. 

Next, we consider equations defined at the vertices of an elementary quadrilateral (see Figure \ref{elementarySQR}), 
\begin{equation}\label{Q-eq}
Q(w,w_{10},w_{01},w_{11};a,b)=0.
\end{equation}

\begin{figure}[ht]
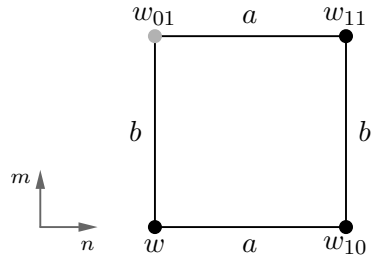

\centering
\centertexdraw{ 
\setunitscale 0.5
\move (-2.4 0)  \lvec(-.4 0) \lvec(-.4 2) \lvec(-2.4 2) \lvec(-2.4 0)
\move (-2.4 0) \fcir f:0.0 r:0.075 \move (-.4 0) \fcir f:0.0 r:0.075
\move (-.4 2) \fcir f:.0 r:0.075 \move (-2.4 2) \fcir f:0.7 r:0.075
\textref h:C v:C \htext(-2.4 -0.2){$w$}
\textref h:C v:C \htext(-.4 -0.2){$w_{10}$}
\textref h:C v:C \htext(-2.4 2.2){$w_{01}$}
\textref h:C v:C \htext(-.4 2.2){$w_{11}$}
\textref h:C v:C \htext(-1.4 -0.2){$a$}
\textref h:C v:C \htext(-1.4 2.2){$a$}
\textref h:C v:C \htext(-2.6 1){$b$}
\textref h:C v:C \htext(-0.2 1){$b$}

\arrowheadsize l:.20 w:.10
\move(-3.6 0) \linewd 0.02 \setgray 0.4 \arrowheadtype t:F \avec(-3 0)
\move(-3.6 0) \linewd 0.02 \setgray 0.4 \arrowheadtype t:F \avec(-3.6 .6)
\htext (-3.1 -.2) {\scriptsize{$n$}}
\htext (-3.8 .5) {\scriptsize{$m$}}
}
\caption{Quad-Graph.}\label{elementarySQR}
\end{figure}
A Lax representation of (\ref{Q-eq}) is an equation 
\begin{equation} \label{Laxintro}
{L}(w_{01},w_{11},a){M}(w,w_{01},b)={M}(w_{10},w_{11},b) {L}(w,w_{10},a),
\end{equation}
for a pair of matrices $L,M$, equivalent to (\ref{Q-eq}). In many cases, as in the one we focus in this paper, $L=M$.  
The three dimensional consistency property is a sufficient (but not necessary) condition for a quad-graph equation to possess a Lax representation 
\cite{Bobenko-Suris,Frank4}. 
In the case of a lift of a quad-graph equation (or system of equations) to a Yang-Baxter map, as presented in \cite{Pap-Tongas}, it can be proven that the Lax matrix of the Yang-Baxter map coincides with the Lax matrix (for $L=M$ in (\ref{Laxintro})) of the corresponding equation \cite{Kouloukas-Dihn}. In many cases, as in the case we investigate in this paper, the converse also holds. That is, the Lax matrix of the initial equation is also a Lax matrix of its lift.  
 As in the case of Yang-Baxter maps, we can consider quadrilateral equations and Lax representations defined on Grassmann variables.

\section{Noncommutative dpKdV lift}\label{dpKdV-lift}
In \cite{kouloukas} it was shown that the following symplectic Yang-Baxter map
\begin{equation}\label{YBliftKdV}
(x_1,x_2,y_1,y_2)\stackrel{Y_{a,b}}{\longmapsto}(u_1,u_2,v_1,v_2)=\left(y_1+\frac{a-b}{x_1-y_2},y_2,x_1,x_2+\frac{a-b}{x_1-y_2}\right),
\end{equation}
can be squeezed down to the well-celebrated dpKdV equation 
\begin{equation}\label{pkdv}
(f_{01}-f_{10})(f-f_{11})=a-b.
\end{equation}
This result was based on the simple observation that $y_1=x_2$ implies $u_1=v_2$.  In particular, if we set $x_2=y_1=f_{10}$, $v_1 =x_1 =f$, $u_2 = y_2 = f_{11}$ 
then the equation $u_1 = v_2 = f_{01}$ is equivalent to (\ref{pkdv}). 
This map was derived as a limit of a non-degenerate Yang-Baxter map, and it satisfies the Lax equation  
 $L_a(u_1,u_2)L_b(v_1,v_2)=L_b(y_1,y_2)L_a(x_1,x_2)$, where 
\begin{equation*}
L_a(\textbf{x}):=
\left(
\begin{matrix}
 x_1 & a+x_1x_2-\lambda \\
 -1 & -x_2
\end{matrix}\right), \quad \textbf{x}:=(x_1,x_2). 
\end{equation*}
At the same time, $L_a$ is a Lax matrix of the dpKdV equation, namely equation ({\ref{pkdv}) is equivalent to 
$$L_a(f_{01},f_{11})L_b(f,f_{01})=L_b(f_{10},f_{11})L_a(f,f_{10}).$$

In general, the Lax matrices of Yang-Baxter maps and quadrilateral equations are not unique. For instance, it is obvious that, the following matrix 
\begin{equation*}
\tilde{L}_a(\textbf{x}):=
\left(
\begin{matrix}
 x_1 & a+x_1x_2-\lambda & 0 \\
 -1 & -x_2 & 0\\
  0 & 0 & 1
\end{matrix}\right),\ \quad \textbf{x}:=(x_1,x_2), 
\end{equation*}
is also a Lax matrix for both the Yang-Baxter map (\ref{YBliftKdV}) and the dpKdV equation. 

Our basic aim is to extend the Yang-Baxter map (\ref{YBliftKdV}) and, consequently, the dpKdV equation to a Grassmann case, by adding two more 
even variables $\chi_1$ and $\chi_2$ as entries of the Lax matrix. Thus, following \cite{GKM, Georgi-Sasha, Sokor-Sasha-2}, we consider the matrix 
\begin{equation}\label{G-laxYBKdV}
\mathcal{L}_a(\textbf{x},\pmb{\chi}):=
\left(
\begin{matrix}
 x_1 & a+x_1x_2+\chi_1\chi_2-\lambda & \chi_1 \\
 -1 & -x_2 & 0\\
  0 & \chi_2 & 1
\end{matrix}\right),\quad (\textbf{x},\pmb{\chi}):=(x_1,x_2,\chi_1,\chi_2),
\end{equation}
where $x_i\in G_0$, $\chi_i\in G_1$, $i=1,2$. This matrix satisfies the following two basic properties. Its bosonic limit is the Lax matrix $\tilde{L}_a(\textbf{x})$, i.e. 
$ \lim_{\pmb{\chi}\rightarrow0}\mathcal{L}_a(\textbf{x},\pmb{\chi})=\tilde{L}_a(\textbf{x})$, and 
$\sdet(\mathcal{L}_a(\textbf{x},\pmb{\chi}))=\det(\tilde{L}_a(\textbf{x}))=a-\lambda$.  We have to notice that this is not the only possible 
extension of $\tilde{L}_a$ satisfying these two properties. However, in this paper, we shall focus on the case of matrix \eqref{G-laxYBKdV}, which constitutes a Lax matrix of 
a noncommutative Yang-Baxter map and its corresponding quadrilateral system. 

In particular, we begin with the following lemma.  
\vspace{1cm}

\begin{lemma}\label{lemma}
The Lax equation
\begin{equation}\label{LaxEq-YB}
\mathcal{L}_a(\textbf{u},\pmb{\xi})\mathcal{L}_b(\textbf{v},\pmb{\eta})=\mathcal{L}_b(\textbf{y},\pmb{\psi})\mathcal{L}_a(\textbf{x},\pmb{\chi}),
\end{equation}
where $\mathcal{L}_a=\mathcal{L}_a(\textbf{x},\pmb{\chi})$ is given by \eqref{G-laxYBKdV} is equivalent to the following system 
\begin{subequations}\label{corresp}
\begin{align}
u_1&=y_1+\frac{a-b}{x_1-y_2+\chi_1\psi_2}, \ 
\xi_1 =\psi_1-\frac{a-b}{x_1-y_2}\chi_1, \ 
\xi_2=\psi_2 \nonumber \\
v_2&=x_2+\frac{a-b}{x_1-y_2+\chi_1\psi_2}, \ 
\eta_1=\chi_1, \ 
\eta_2=\chi_2+\frac{a-b}{x_1-y_2}\psi_2 \nonumber 
\end{align}
\end{subequations}
and 
\begin{equation}\label{equ2v1}
u_2-v_1=y_2-x_1
\end{equation}
\end{lemma}
\begin{proof}
Equation \eqref{LaxEq-YB} is equivalent to $\xi_2=\psi_2$, $\eta_1=\chi_1$, equation \eqref{equ2v1} and the following system
\begin{subequations}\label{LaxSys}
\begin{align}
&v_2-u_1=x_2-y_1,\label{LaxSys-a}\\
&u_1v_1-(a+u_1u_2+\xi_1\psi_2)=y_1x_1-(b+y_1y_2+\psi_1\psi_2),\label{LaxSys-b}\\
&u_2v_2-(b+v_1v_2+\chi_1\eta_2)=y_2x_2-(a+x_1x_2+\chi_1\chi_2),\label{LaxSys-c}\\
&\xi_1+u_1\chi_1=\psi_1+y_1\chi_1,\label{LaxSys-d}\\
&\eta_2-\psi_2v_2=\chi_2-\psi_2x_2,\label{LaxSys-e}\\
&u_1(b+v_1v_2+\chi_1\eta_2)-v_2(a+u_1u_2+\xi_1\psi_2)+\xi_1\eta_2=\nonumber\\
&y_1(a+x_1x_2+\chi_1\chi_2)-x_2(b+y_1y_2+\psi_1\psi_2)+\psi_1\chi_2,\label{LaxSys-f}
\end{align}
\end{subequations}
for variables $u_1$, $u_2$, $\xi_1$, $v_1$, $v_2$ and $\eta_2$. Now, from \eqref{LaxSys-d} follows that $\xi_1\psi_2=\psi_1\psi_2+(y_1-u_1)\chi_1\psi_2$. Substituting the latter into \eqref{LaxSys-b} and using \eqref{equ2v1}, we obtain
\begin{equation*}
u_1(x_1-y_2+\chi_1\psi_2)=a-b+y_1(x_1-y_2+\chi_1\psi_2),
\end{equation*}
which implies $$u_1=y_1+\frac{a-b}{x_1-y_2+\chi_1\psi_2}.$$ Consequently, using \eqref{LaxSys-a}, we derive $v_2=x_2+\frac{a-b}{x_1-y_2+\chi_1\psi_2}$. In addition, equations \eqref{LaxSys-d} and \eqref{LaxSys-e} imply the following
\begin{equation}\label{xi1eta2unsimplified}
\xi_1=\psi_1+\frac{b-a}{x_1-y_2+\chi_1\psi_2}\chi_1, \ \eta_2=\chi_2+\frac{a-b}{x_1-y_2+\chi_1\psi_2}\psi_2,
\end{equation}
respectively. Multiplying both numerators and denominators of the above fractions with the conjugate expression of the denominator ``$x_1-y_2-\chi_1\psi_2$", and taking into account the fact that $\chi_1^2=\psi_2^2=0$ (since $\chi_1,\psi_2\in G_0$), we obtain $\xi_1$ and $\eta_2$.
\end{proof}

Since equation \eqref{equ2v1} is valid for any value of $u_2,v_1\in G_0$, system \eqref{corresp}-\eqref{equ2v1} is not defining a map, but a correspondence. Yet, this correspondence does not satisfy the Yang-Baxter equation for any $u_2,v_1\in G_0$ satisfying \eqref{equ2v1}. Nevertheless, as the following theorem states, a particular solution of \eqref{equ2v1} defines a Yang-Baxter map.


\begin{theorem}\label{YBproperty}
The map
\begin{equation}\label{GmapYB}
S_{a,b}:((\textbf{x},\pmb{\chi}),(\textbf{y},\pmb{\psi})){\mapsto} ((\textbf{u},\pmb{\xi}),(\textbf{v},\pmb{\eta})),
\end{equation}
defined by
\begin{subequations}\label{superYB}
\begin{align}
x_1\mapsto u_1&=y_1+\frac{a-b}{x_1-y_2+\chi_1\psi_2};\label{superYB-a}\\
x_2\mapsto u_2&=y_2;\label{superYB-b}\\
\chi_1\mapsto \xi_1 &=\psi_1-\frac{a-b}{x_1-y_2}\chi_1;\label{superYB-c}\\
\chi_2\mapsto \xi_2&=\psi_2;\label{superYB-d}\\
y_1\mapsto v_1&=x_1;\label{superYB-e}\\
y_2\mapsto v_2&=x_2+\frac{a-b}{x_1-y_2+\chi_1\psi_2};\label{superYB-f}\\
\psi_1\mapsto \eta_1&=\chi_1;\label{superYB-g}\\
\psi_2\mapsto \eta_2&=\chi_2+\frac{a-b}{x_1-y_2}\psi_2,\label{superYB-h}
\end{align}
\end{subequations}
where $x_i,y_i,u_i,v_i\in G_0$, $\chi_i,\psi_i,\xi_i,\eta_i\in G_1$, $i=1,2$, is a parametric Yang-Baxter map, with Lax matrix 
$\mathcal{L}_a$ given by \eqref{G-laxYBKdV}, that satisfies 
\begin{equation}\label{Laxrept}
\mathcal{L}_a(u_1,u_2,\xi_1,\xi_2)\mathcal{L}_b(v_1,v_2,\eta_1,\eta_2)=\mathcal{L}_b(y_1,y_2,\psi_1,\psi_2)\mathcal{L}_a(x_1,x_2,\chi_1,\chi_2).
\end{equation}
Furthermore, $S_{a,b}$ admits the following invariants
\begin{subequations}\label{sinvYB}
\begin{align}
I_1&=(x_1-y_2)(y_1-x_2),\label{sinvYB-a} \\ 
I_2&=\chi_1\chi_2+\psi_1\psi_2,\label{sinvYB-b} \\ 
I_3&=x_2-x_1+y_2-y_1,\label{sinvYB-c} \\
I_4&=(x_1-y_2)\chi_2\psi_1+(x_2-y_1)\chi_1\psi_2-a\psi_1\psi_2-b\chi_1\chi_2,\label{sinvYB-d}
\end{align}
\end{subequations}
and its bosonic limit is the lift of the dpKdV equation \eqref{YBliftKdV}.

\end{theorem}
\begin{proof}
See Appendix.
\end{proof}

\section{Squeeze down to a noncommutative dpKdV system}\label{lift-down}

In this section, we show that the Yang-Baxter map of theorem \ref{YBproperty} can be squeezed down to a system of equations on a quad-graph, which possesses 
a Lax representation. This system is defined with three fields on each vertex of an elementary quadrilateral, one even and two odd.  Its bosonic limit leads to the 
dpKdV equation. 

Similarly to \cite{Kouloukas-Dihn, Pap-Tongas}, we prove the following proposition for a specific class of Yang-Baxter maps.  

\begin{proposition}\label{map4graph}
Let $S_{a,b}:\left((\textbf{x},\pmb{\chi}),(\textbf{y},\pmb{\psi})\right) \mapsto \left((\textbf{u},\pmb{\xi}),(\textbf{v},\pmb{\eta})\right)$ be a parametric Yang-Baxter map with
\begin{subequations}\label{YBmapForm}
\begin{align}
&(u_1,u_2,\xi_1,\xi_2):=\left(F_{a,b}(y_1,x_1,y_2,\chi_1,\psi_2),y_2,\Phi^1_{a,b}(x_1,y_2,\psi_1,\chi_1),\psi_2\right)\label{YBmapForm-b}\\
&(v_1,v_2,\eta_1,\eta_2):=\left(x_1,F_{a,b}(x_2,x_1,y_2,\chi_1,\psi_2),\chi_1,\Phi^2_{a,b}(x_1,y_2,\chi_2,\psi_2)\right),\label{YBmapForm-c}
\end{align}
\end{subequations}
and associated Lax matrix $\mathcal{L}_a=\mathcal{L}_a(x_1,x_2,\chi_1,\chi_2)$, where $x_i\in G_0$ and $\chi_i\in G_1$, $i=1,2$. Then, the following system of equations
\begin{align}
f_{01}&=F_{a,b}(f_{10},f,f_{11},\tau,\phi_{11}),\nonumber\\
\phi_{01}&=\Phi^2_{a,b}(f,f_{11},\phi_{10},\phi_{11}), \label{mapQuadGraph}\\
\tau_{01}&=\Phi^1_{a,b}(f,f_{11},\tau_{10},\tau),\nonumber
\end{align}
satisfies the Lax equation
\begin{equation}\label{dLaxRep}
\mathcal{L}_a(f_{01},f_{11},\tau_{01},\phi_{11})\mathcal{L}_b(f,f_{01},\tau,\phi_{01})=\mathcal{L}_b(f_{10},f_{11},\tau_{10},\phi_{11})\mathcal{L}_a(f,f_{10},\tau,\phi_{10}).
\end{equation}
\end{proposition}
\begin{proof}
Setting $x_2=y_1$, equations \eqref{YBmapForm-b}-\eqref{YBmapForm-c} imply $u_1=v_2$. Now, labelling $v_1=x_1=f$, $x_2=y_1=f_{10}$, $u_2=y_2=f_{11}$, $\eta_1=\chi_1=\tau$, $\chi_2=\phi_{10}$, $\psi_1=\tau_{10}$ and $\xi_2=\psi_2=\phi_{11}$, then equations \eqref{YBmapForm} and \eqref{mapQuadGraph} imply $u_1=v_2=f_{01}$, $\xi_1=\tau_{01}$ and $\eta_2=\phi_{01}$. Finally, substituting these values to \eqref{Laxrept}, we obtain \eqref{dLaxRep}.
\end{proof}

\begin{remark}
We can interpret system (\ref{mapQuadGraph}) with its fields lying on the vertices of the quad-graph of Figure \ref{elementarySQR}, namely the vector variables $w=(f,\tau,\phi)$,  
$w_{10}=(f_{10},\tau_{10},\phi_{10})$, $w_{01}=(f_{01},\tau_{01},\phi_{01})$ and $w_{11}=(f_{11},\tau_{11},\phi_{11})$ are being placed on each corresponding vertex. The Lax representation  
\eqref{dLaxRep} of this system is equivalent 
to (\ref{Laxintro}), for $M=L$ and 
$$L(w,w_{10},a)=L \left( (f,\tau,\phi), (f_{10},\tau_{10},\phi_{10}),a \right)\equiv \mathcal{L}_a(f,f_{10},\tau,\phi_{10}).$$
\end{remark}

Now, in the case of the Yang-Baxter map of theorem \ref{YBproperty}, the above proposition leads to the following. 
\begin{theorem}(Grassmann extension of dpKdV)
The following system 
\begin{align}
&(f_{01}-f_{10})(f-f_{11}+\tau\phi_{11})=a-b,\nonumber\\
&(\phi_{01}-\phi_{10})(f-f_{11})=(a-b)\phi_{11},\label{G-dpKdV}\\
&(\tau_{01}-\tau_{10})(f-f_{11})=(b-a)\tau.\nonumber
\end{align}
admits the Lax representation \eqref{dLaxRep} with Lax matrix
\begin{equation}\label{G-laxdpKdV}
 \mathcal{L}_a(f,f_{10},\tau,\phi_{10}):=
\left(
\begin{matrix}
 f & a+f f_{10}+\tau \phi_{10}-\lambda & \tau \\
 -1 & -f_{10} & 0\\
  0 & \phi_{10} & 1
\end{matrix}\right).
\end{equation}
Its bosonic limit is the dpKdV equation \eqref{pkdv}. 
\end{theorem}
\begin{proof}
Due to Proposition \ref{map4graph}, from \eqref{superYB-a}, \eqref{superYB-c} and \eqref{superYB-h}, we obtain the following system of equations
\begin{align*}
f_{01}&=f_{10}+\frac{a-b}{f-f_{11}+\tau \phi_{11}},\\
\phi_{01}&=\phi_{10}+\frac{a-b}{f-f_{11}}\phi_{11},\\
\tau_{01}&=\tau_{10}-\frac{a-b}{f-f_{11}}\tau,
\end{align*}
which can be rewritten in the form of system \eqref{G-dpKdV}. The Lax matrix \eqref{G-laxdpKdV} is derived from \eqref{G-laxYBKdV} by setting $(x_1,x_2,\chi_1,\chi_2)\rightarrow (f,f_{10},\tau,\phi_{10})$. Finally, by setting $\tau=\tau_{10}=\tau_{01}=0$ and $\phi_{10}=\phi_{01}=\phi_{11}=0$ in \eqref{G-dpKdV}, it follows that $f$ satisfies the classical dpKdV equation.
\end{proof}

\section{Integrable commutative analogues} \label{commutative}
Given the map \eqref{GmapYB}, the question arises as to weather it preserves the Yang-Baxter property, if we consider all variables in \eqref{superYB} to be commutative. Although the answer is negative,   if we consider the version of the map before simplifying the entries $\xi_1$ and $\eta_2$ in \eqref{superYB-c} and \eqref{superYB-h}, respectively, using the properties of odd variables (see \eqref{xi1eta2unsimplified}), then the corresponding map with all its variables being considered as even satisfies the Yang-Baxter equation. That is, the following map
\begin{equation}\label{commYBmap}
Y_{a,b}:(\textbf{x},\textbf{y}){\rightarrow}(\textbf{u},\textbf{v}),
\end{equation}
which is defined by
\begin{subequations}\label{commYB}
\allowdisplaybreaks
\begin{align}
x_1\mapsto u_1&=y_1+\frac{a-b}{x_1-y_2+X_1 Y_2};\\
x_2\mapsto u_2&=y_2;\\
X_1\mapsto U_1 &=Y_1-\frac{a-b}{x_1-y_2+X_1Y_2}X_1\\
X_2\mapsto U_2 &=Y_2\\
y_1\mapsto v_1&=x_1;\\
y_2\mapsto v_2&=x_2+\frac{a-b}{x_1-y_2+X_1Y_2};\\
Y_1\mapsto V_1&=X_1;\label{commYB-g}\\
Y_2\mapsto V_2&=X_2+\frac{a-b}{x_1-y_2+X_1Y_2}Y_2,
\end{align}
\end{subequations}
is a parametric Yang-Baxter map, and possesses the following (not strong) Lax representation:
\begin{equation*}
L_a(\textbf{u})L_b(\textbf{v})=L_b(\textbf{y})L_a(\textbf{x}),
\end{equation*}
where
\begin{equation} \label{comLax}
L_a(\textbf{x}):=\left(
\begin{matrix}
 x_1 & a+x_1x_2+X_1 X_2-\lambda & X_1 \\
 -1 & -x_2 & 0\\
  0 & X_2 & 1
\end{matrix}\right),
\end{equation}
and $\textbf{x}:=(x_1,x_2,X_1,X_2)$. 

By direct calculation, we can  prove that this Lax matrix is associated with the Sklyanin bracket \cite{skly1, skly2}. Specifically, we have the following.

\begin{proposition}
Matrix $L_a(\textbf{x})=L_a(\textbf{x},\lambda)$ in \eqref{comLax} satisfies 
\begin{equation} \label{sklyanin}
\{L(\mathbf{x},\lambda_1 ) \ \overset{\otimes }{,} \
L(\mathbf{x},\lambda_2)\}=[\frac{r}{\lambda_1-\lambda_2}(\lambda_1,\lambda_2),L(\mathbf{x},\lambda_1 )\otimes
L(\mathbf{x},\lambda_2)], 
\end{equation}
with $r(x\otimes y) = y\otimes x$, iff $\{x_1,x_2\}=\{X_1,X_2\}=1$ and $\{x_1,X_1\}=\{x_1,X_2\}=\{x_2,X_1\}=\{x_2,X_2\}=0$. 
The map \eqref{commYBmap} is Poisson with respect to the corresponding extended bracket in 
$V \times V$, i.e. $$ \pi=\frac{\partial}{\partial x_1} \wedge
\frac{\partial}{\partial x_2}+\frac{\partial}{\partial X_1} \wedge
\frac{\partial}{\partial X_2}+\frac{\partial}{\partial y_1} \wedge
\frac{\partial}{\partial y_2}+\frac{\partial}{\partial Y_1} \wedge
\frac{\partial}{\partial Y_2}.$$
\end{proposition}

It can be readily verified that map \eqref{commYBmap}-\eqref{commYB} admits the following functionally independent integrals 
\begin{align*}
I_1&=(x_1-y_2)(x_2-y_1),\\
I_2&=X_1X_2+Y_1Y_2,\\
I_3&=x_2-x_1+y_2-y_1,\\
I_4&=(x_1-y_2+X_1Y_2)(x_2-y_1-X_2Y_1)-aY_1Y_2-bX_1X_2,
\end{align*}
which are in involution with respect to the Poisson bracket $\pi$. Thus, map \eqref{commYBmap}-\eqref{commYB} is completely integrable in the Liouville sense.

Given a Yang-Baxter map, we are mostly interested in studying the integrability of the corresponding transfer maps which occur by considering periodic initial value problems on quadrilateral lattices \cite{kp2, PNC,Veselov,Veselov3}.  In this framework, the Sklyanin bracket provides an efficient approach to prove the Liouville integrability. The comultiplication property of the Sklyanin bracket states that the monodromy matrix will also satisfy equation \eqref{sklyanin}, which, in our case, is equivalent to the canonical Poisson structure in $(V \times V)^n$, and the corresponding transfer maps preserve this Poisson structure. Furthermore, the Sklyanin bracket ensures the involutivity of the integrals derived by the trace of the corresponding monodromy matrix, which consists of products of local Lax matrices.
 However, we have to mention that, in this particular case, the spectrum of the monodromy matrix does not provide enough integrals to claim the Liouville integrability of the transfer maps. Even in the simplest case of transfer map, namely the map $Y_{a,b}$ itself (one periodic problem), we derive only three integrals from the spectrum of the corresponding  monodromy matrix $M=L_b(\textbf{y})L_a(\textbf{x})$. That is, integrals $I_1,I_2$ and $I_4$, while $I_3$ was derived simply by inspection. For these kind of integrals, as the latter, we have to prove directly that they are in involution with the rest.  
 
In a similar way as in the noncommutative case, map \eqref{commYBmap}-\eqref{commYB} can be squeezed down to the following system
\begin{align}
(f_{01}-f_{10})(f-f_{11}+g h_{11})&=a-b,\nonumber\\
(h_{01}-h_{10})(f-f_{11}+g h_{11})&=(a-b)h_{11},  \label{comsystem} \\
(g_{01}-g_{10})(f-f_{11}+g h_{11})&=(b-a)g, \nonumber
\end{align}
which possesses the Lax representation 
\begin{equation*}
L_a(f_{01},f_{11},g_{01},h_{11})L_b(f,f_{01},g,h_{01})=L_b(f_{10},f_{11},g_{10},h_{11})L_a(f,f_{10},g,h_{10}),
\end{equation*}
with Lax matrix (derived from the corresponding Lax matrix of $Y_{a,b}$)  
\begin{equation} \label{comsysLax}
L_a(f,f_{10},g,h_{10}):=
\left(
\begin{matrix}
 f & a+f f_{10}+g h_{10}-\lambda & g \\
 -1 & -f_{10} & 0\\
  0 & h_{10} & 1
\end{matrix}\right).
\end{equation}

\begin{remark}\normalfont
The Lax representations of Yang-Baxter maps  and of quad-graph equations are invariant under conjugation. In this sense, the Lax matrix (\ref{comsysLax}) is equivalent to  
\begin{equation} \label{conj}
P^{-1} L_a(f,f_{10},g,h_{10}) P=
\left(
\begin{matrix}
 -f_{10} & 0& -1 \\
 h_{10} & 1 & 0\\
 a+f f_{10}+g h_{10}-\lambda& g & f
\end{matrix}\right),
\end{equation}  
where $P=\left(
\begin{matrix}
0 & 0& 1 \\
 1 & 0 & 0\\
 0 & 1 & 0
\end{matrix}\right)$.
This Lax matrix can be seen in comparison with the Lax matrix of the discrete Boussinesq system \cite{Brid, NiPCQ, Tongas-Frank}. 
The main difference is that the $2 \times 2$ up-right block of the Boussinesq Lax matrix is the identity matrix instead of the symplectic block that appears in (\ref{conj}). 
This crucial difference indicates that 
the Lax matrices and the corresponding systems are not equivalent. Similarly, there is a Yang-Baxter map (not equivalent to (\ref{commYB})), associated with the Boussinesq system. 
\end{remark}

\section{Conclusions}  \label{conclusions}
We constructed a Grassmann extension of a Yang-Baxter map, namely map \eqref{GmapYB}-\eqref{superYB}, from a ``lift'' of the famous dpKdV equation, and we proved that it possesses the Yang-Baxter property. We showed that this map can be squeezed down to a Grassmann quad-graph dpKdV system, i.e. system \eqref{mapQuadGraph}, 
with three fields on each vertex of the quad-graph; one even and two odd. Both our Yang-Baxter map \eqref{GmapYB}-\eqref{superYB} and system \eqref{mapQuadGraph} have quite simple form and possess  zero curvature representations with $3 \times 3$ binomial (with respect to the spectral parameter $\lambda$) Lax matrices. Finally, we studied the commutative analogues of these constructions, and we proved that they satisfy similar properties. 

We believe that our results can be extended in several ways. In particular, something that deserves further investigation is  the relation of map \eqref{GmapYB}-\eqref{superYB} with integrable PDEs. 
In \cite{GKM, Sokor-Sasha-2}, the associated Lax matrices of the derived Yang-Baxter maps constitute Darboux transformations of certain PDEs of NLS and generalised KdV type, respectively. We believe that, if the entries of Lax matrix \eqref{G-laxYBKdV} are viewed as potential functions, then matrix \eqref{G-laxYBKdV} may constitute a Darboux transformation of some (noncommutative) integrable system of PDEs. 

Moreover, we believe that similar Grassmann extensions can be constructed for various integrable quad-graph equations and their corresponding Yang-Baxter maps. It would be 
interesting to investigate this problem in the case of the equations of the  ABS  classification list.

Finally, regarding the commutative analogues presented in the last section, we believe that a systematic study of the Liouville integrability of periodic reductions 
of system \eqref{comsystem}, as well as of higher dimensional transfer maps of the Yang-Baxter map \eqref{commYBmap}, deserves further investigation. 
It seems that the Sklyanin bracket provides the right framework to deal with this kind of problems \cite{Kouloukas-Dihn}.

\section*{Acknowledgements}
We would like to thank Profs A.N.W Hones, A.V. Mikhailov, V.G. Papageorgiou and J.P. Wang for the discussion and their useful comments. T.K. acknowledges support by EPSRC (Grant EP/M004333/1). The main part of this paper was written during the SIDE 12 conference in Montreal, thus we would like to thank the organisers for the opportunity to participate in this meeting; S.K.R. would also like to thank the organisers for the financial support.
 
\appendix
\section{Proof of theorem \ref{YBproperty}}
The defined $(u_1,\xi_1,v_2,\eta_2)$ satisfy equations \eqref{LaxSys} for any $u_2,v_1\in G_1$. We choose $u_2=y_2$, which implies $v_1=x_1$. Due to Lemma \ref{lemma}, map $\{\eqref{GmapYB}-\eqref{superYB}\}$ satisfies the Lax equation  \eqref{Laxrept}. However, using \eqref{superYB-b}, \eqref{superYB-d}, \eqref{superYB-e} and \eqref{superYB-g}, the Lax equation \eqref{Laxrept} reads
\begin{equation}\label{LaxreptApp}
\mathcal{L}_a(u_1,y_2,\xi_1,\psi_2)\mathcal{L}_b(x_1,v_2,\chi_1,\eta_2)=\mathcal{L}_b(y_1,y_2,\psi_1,\psi_2)\mathcal{L}_a(x_1,x_2,\chi_1,\chi_2).
\end{equation}

Now, let $((x_i,\chi_i),(y_i,\psi_i),(z_i,\zeta_i))\in V_G\times V_G\times V_G$, $i=1,2$. Then, using the left side of the Yang-Baxter equation, we adopt the following notation.
\begin{align*}
S^{23}_{b,c}(x_1,x_2,\chi_1,\chi_2,y_1,y_2,\psi_1,\psi_2,z_1,z_2,\zeta_1,\zeta_2)&=(x_1,x_2,\chi_1,\chi_2,\tilde{y}_1,\tilde{y}_2,\tilde{\psi}_1,\tilde{\psi}_2,\tilde{z}_1,\tilde{z}_2,\tilde{\zeta_1},\tilde{\zeta_2});\\
S^{13}_{a,c}\circ S^{23}_{b,c}(x_1,x_2,\chi_1,\chi_2,y_1,y_2,\psi_1,\psi_2,z_1,z_2,\zeta_1,\zeta_2)&=(\tilde{x}_1,\tilde{x}_2,\tilde{\chi_1},\tilde{\chi}_2,\tilde{y}_1,\tilde{y}_2,\tilde{\psi}_1,\tilde{\psi}_2,\tilde{\tilde{z}}_1,\tilde{\tilde{z}}_2,\tilde{\tilde{\zeta}}_1,\tilde{\tilde{\zeta}}_2);\\
S^{12}_{a,b}\circ S^{13}_{a,c}\circ S^{23}_{b,c}(x_1,x_2,\chi_1,\chi_2,y_1,y_2,\psi_1,\psi_2,z_1,z_2,\zeta_1,\zeta_2)&=(\tilde{\tilde{x}}_1,\tilde{\tilde{x}}_2,\tilde{\tilde{\chi}}_1,\tilde{\tilde{\chi}}_2,\tilde{\tilde{y}}_1,\tilde{\tilde{y}}_2,\tilde{\tilde{\psi}}_1,\tilde{\tilde{\psi}}_2,\tilde{\tilde{z}}_1,\tilde{\tilde{z}}_2,\tilde{\tilde{\zeta}}_1,\tilde{\tilde{\zeta}}_2).
\end{align*}
Now, using the right side of the Yang-Baxter equation, 
\begin{align*}
S^{12}_{a,b}(x_1,x_2,\chi_1,\chi_2,y_1,y_2,\psi_1,\psi_2,z_1,z_2,\zeta_1,\zeta_2)&=(\hat{x}_1,\hat{x}_2,\hat{\chi}_1,\hat{\chi}_2,\hat{y}_1,\hat{\psi}_1,\hat{\psi}_2,\hat{y}_2,z_1,z_2,\zeta_1,\zeta_2);\\
S^{13}_{a,c}\circ S^{12}_{a,b}(x_1,x_2,\chi_1,\chi_2,y_1,y_2,\psi_1,\psi_2,z_1,z_2,\zeta_1,\zeta_2)&=(\hat{\hat{x}}_1,\hat{\hat{x}}_2,\hat{\hat{\chi}}_1,\hat{\hat{\chi}}_2,\hat{y}_1,\hat{y}_2,\hat{\psi}_1,\hat{\psi}_2,\hat{z}_1,\hat{z}_2,\hat{\zeta}_1,\hat{\zeta}_2);\\
S^{23}_{b,c}\circ S^{13}_{a,c}\circ S^{12}_{a,b}(x_1,x_2,\chi_1,\chi_2,y_1,y_2,\psi_1,\psi_2,z_1,z_2,\zeta_1,\zeta_2)&=(\hat{\hat{x}}_1,\hat{\hat{x}}_2,\hat{\hat{\chi}}_1,\hat{\hat{\chi}}_2,\hat{\hat{y}}_1,\hat{\hat{y}}_2,\hat{\hat{\psi}}_1,\hat{\hat{\psi}}_2,\hat{\hat{z}}_1,\hat{\hat{z}}_2,\hat{\hat{\zeta}}_1,\hat{\hat{\zeta}}_2).
\end{align*}

In this notation, applying the left part of the Yang-Baxter equation \eqref{SYB_eq1} to the product $\mathcal{L}_a(x_1,x_2,\chi_1,\chi_2)$ \allowbreak $\mathcal{L}_b(y_1,y_2,\psi_1,\psi_2)\mathcal{L}_c(z_1,z_2,\zeta_1,\zeta_2)$, and using \eqref{LaxreptApp} consecutively, it follows that:
\begin{align}
\mathcal{L}_c(z_1,z_2,\zeta_1,\zeta_2)&\mathcal{L}_b(y_1,y_2,\psi_1,\psi_2)\mathcal{L}_a(x_1,x_2,\chi_1,\chi_2)=\nonumber\\
&=\mathcal{L}_c(z_1,z_2,\zeta_1,\zeta_2)\mathcal{L}_a(\tilde{x}_1,y_2,\tilde{\chi}_1,\psi_2)\mathcal{L}_b(x_1,\tilde{y}_2,\chi_1,\tilde{\psi}_2)\nonumber\\
&=\mathcal{L}_a(\tilde{\tilde{x}}_1,z_2,\tilde{\tilde{\chi}}_1,\zeta_2)\mathcal{L}_c(\tilde{x}_1,\tilde{z}_2,\tilde{\chi}_1,\tilde{\zeta}_2)\mathcal{L}_b(x_1,\tilde{y}_2,\chi_1,\tilde{\psi}_2)\nonumber\\
&=\mathcal{L}_a(\tilde{\tilde{x}}_1,z_2,\tilde{\tilde{\chi}}_1,\zeta_2)\mathcal{L}_b(\tilde{x}_1,\tilde{z}_2,\tilde{\chi}_1,\tilde{\zeta}_2)\mathcal{L}_c(x_1,\tilde{\tilde{z}}_2,\chi_1,\tilde{\tilde{\zeta}}_2)\label{3facYB-1}.
\end{align}
On the other hand, applying the right side of equation \eqref{SYB_eq1}:
\begin{align}
\mathcal{L}_c(z_1,z_2,\zeta_1,\zeta_2)&\mathcal{L}_b(y_1,y_2,\psi_1,\psi_2)\mathcal{L}_a(x_1,x_2,\chi_1,\chi_2)=\nonumber\\
&=\mathcal{L}_b(\hat{y}_1,z_2,\hat{\psi}_1,\zeta_2)\mathcal{L}_c(y_1,\hat{z}_2,\psi_1,\hat{\zeta}_2)\mathcal{L}_a(x_1,x_2,\chi_1,\chi_2)\nonumber\\
&=\mathcal{L}_b(\hat{y}_1,z_2,\hat{\psi}_1,\zeta_2)\mathcal{L}_a(\hat{x}_1,\hat{z}_2,\hat{\chi}_1,\hat{\zeta}_2)\mathcal{L}_c(x_1,\hat{\hat{z}}_2,\chi_1,\hat{\hat{\zeta}}_2)\nonumber\\
&=\mathcal{L}_a(\hat{\hat{x}}_1,z_2,\hat{\hat{\chi}}_1,\zeta_2)\mathcal{L}_b(\hat{x}_1,\hat{z}_2,\hat{\chi}_1,\hat{\zeta}_2)\mathcal{L}_c(x_1,\hat{\hat{z}}_2,\chi_1,\hat{\hat{\zeta}}_2).\label{3facYB-2}
\end{align}

Comparing \eqref{3facYB-1} and \eqref{3facYB-2}, equation
\begin{equation}\label{uvwEq}
\mathcal{L}_a(\tilde{\tilde{x}}_1,z_2,\tilde{\tilde{\chi}}_1,\zeta_2)\mathcal{L}_b(\tilde{x}_1,\tilde{z}_2,\tilde{\chi}_1,\tilde{\zeta}_2)\mathcal{L}_c(x_1,\tilde{\tilde{z}}_2,\chi_1,\tilde{\tilde{\zeta}}_2)=\mathcal{L}_a(\hat{\hat{x}}_1,z_2,\hat{\hat{\chi}}_1,\zeta_2)\mathcal{L}_b(\hat{x}_1,\hat{z}_2,\hat{\chi}_1,\hat{\zeta}_2)\mathcal{L}_c(x_1,\hat{\hat{z}}_2,\chi_1,\hat{\hat{\zeta}}_2),
\end{equation}
where $\mathcal{L}_a\equiv \mathcal{L}_a(x_1,x_2,\chi_1,\chi_2)$ is given by \eqref{G-laxYBKdV}. We need to show that the matrix refactorisation problem \eqref{uvwEq} implies that
$$
\tilde{\tilde{x}}_1=\hat{\hat{x}}_1,~~ \tilde{\tilde{\chi}}_1=\hat{\hat{\chi}}_1,~~ \tilde{x}_1=\hat{x}_1,~~\tilde{z}_2=\hat{z}_2,~~\tilde{\chi}_1=\hat{\chi}_1,~~\tilde{\zeta}_2=\hat{\zeta}_2,~~\tilde{\tilde{z}}_2=\hat{\hat{z}}_2,~~\text{and}~~\tilde{\tilde{\zeta}}_2=\hat{\hat{\zeta}}_2.
$$

Indeed, equation \eqref{uvwEq} implies the following system of equations:
\begin{subequations}\label{syseqs}
\allowdisplaybreaks
\begin{align}
&\tilde{\tilde{x}}_1-\tilde{z}_2=\hat{\hat{x}}_1-\hat{z}_2,\label{syseqs-a}\\
&\tilde{x}_1-\tilde{\tilde{z}}_2=\hat{x}_1-\hat{\hat{z}}_2,\label{syseqs-b}\\
&\tilde{\chi}_1+\tilde{x}_1\chi_1=\hat{\chi}_1+\hat{x}_1\chi_1,\label{syseqs-c}\\
&\tilde{\zeta}_2-\tilde{z}_2\zeta_2=\hat{\zeta}_2-\hat{z}_2\zeta_2,\label{syseqs-d}\\ 
&\tilde{z}_2(\tilde{x}_1-z_2)-x_1\tilde{x}_1+\tilde{\chi}_1\tilde{\zeta}_2=\hat{z}_2(\hat{x}_1-z_2)-x_1\hat{x}_1+\hat{\chi}_1\hat{\zeta}_2,\label{syseqs-e}\\
&\tilde{\tilde{\chi}}_1+\tilde{\tilde{x}}_1\tilde{\chi}_1+\left[\tilde{\tilde{x}}_1(\tilde{x}_1-z_2)-a-\tilde{\tilde{\chi}}_1\zeta_2\right]\chi_1=
\hat{\hat{\chi}}_1+\hat{\hat{x}}_1\hat{\chi}_1+\left[\hat{\hat{x}}_1(\hat{x}_1-z_2)-a-\hat{\hat{\chi}}_1\zeta_2\right]\chi_1\label{syseqs-f}\\
&\tilde{\tilde{x}}_1(z_2-\tilde{x}_1)+\tilde{\tilde{\chi}}_1\zeta_2+\chi_1\tilde{\tilde{\zeta}}_2+(x_1+\tilde{\tilde{x}}_1-\tilde{z}_2)\tilde{\tilde{z}}_2=\nonumber\\
&\hat{\hat{x}}_1(z_2-\hat{x}_1)+\hat{\hat{\chi}}_1\zeta_2+\chi_1\hat{\hat{\zeta}}_2+(x_1+\hat{\hat{x}}_1-\hat{z}_2)\hat{\hat{z}}_2,\label{syseqs-g}\\
&\tilde{\tilde{x}}_1x_1(\tilde{x}_1-z_2)-\tilde{\tilde{x}}_1(b+\tilde{x}_1\tilde{z}_2+\tilde{\chi}_1\tilde{\zeta}_2)+\tilde{z}_2(a+\tilde{\tilde{x}}_1z_2+\tilde{\tilde{\chi}}_1\zeta_2)+(x_1\zeta_2+\tilde{\zeta}_2)\tilde{\tilde{\chi}}_1=\nonumber\\
&\hat{\hat{x}}_1x_1(\hat{x}_1-z_2)-\hat{\hat{x}}_1(b+\hat{x}_1\hat{z}_2+\hat{\chi}_1\hat{\zeta}_2)+\hat{z}_2(a+\hat{\hat{x}}_1z_2+\hat{\hat{\chi}}_1\zeta_2)+(x_1\zeta_2+\hat{\zeta}_2)\hat{\hat{\chi}}_1,\label{syseqs-h}\\
&(c+x_1\tilde{\tilde{z}}_2+\chi_1\tilde{\tilde{z}}_2)\left[\tilde{\tilde{x}}_1(\tilde{x}_1-z_2)-a-\tilde{\tilde{\chi}}_1\zeta_2\right]+\tilde{\tilde{x}}_1\tilde{z}_2\tilde{\tilde{z}}_2(z_2-\tilde{x}_1)\nonumber\\
&+\tilde{\tilde{z}}_2\left[\tilde{z}_2(a+\tilde{\tilde{\chi}}_1\zeta_2)-\tilde{\tilde{x}}_1(b+\tilde{\chi}_1\tilde{\zeta}_2)-\tilde{\tilde{\chi}}_1\tilde{\zeta}_2\right]+(\tilde{\tilde{\chi}}_1+\tilde{\tilde{x}}_1\tilde{\chi}_1)\tilde{\tilde{\zeta}}_2=\nonumber\\
&(c+x_1\hat{\hat{z}}_2+\chi_1\hat{\hat{\zeta}}_2)\left[\hat{\hat{x}}_1(\hat{x}_1-z_2)-a-\hat{\hat{\chi}}_1\zeta_2\right]+\hat{\hat{x}}_1\hat{z}_2\hat{\hat{z}}_2(z_2-\hat{x}_1)\nonumber\\
&+\hat{\hat{z}}_2\left[\hat{z}_2(a+\hat{\hat{\chi}}_1\zeta_2)-\hat{\hat{x}}_1(b+\hat{\chi}_1\hat{\zeta}_2)-\hat{\hat{\chi}}_1\hat{\zeta}_2\right]+(\hat{\hat{\chi}}_1+\hat{\hat{x}}_1\hat{\chi}_1)\hat{\hat{\zeta}}_2,\label{syseqs-i}\\
&\tilde{\tilde{z}}_2\left[b+(\tilde{z}_2-x_1)(\tilde{x}_1-z_2)+\tilde{\chi}_1\tilde{\zeta}_2\right]-c\tilde{x}_1+\left[(z_2-\tilde{x}_1)\chi_1-\tilde{\chi}_1\right]\tilde{\tilde{\zeta}}_2=\nonumber\\
&\hat{\hat{z}}_2\left[b+(\hat{z}_2-x_1)(\hat{x}_1-z_2)+\hat{\chi}_1\hat{\zeta}_2\right]-c\hat{x}_1+\left[(z_2-\hat{x}_1)\chi_1-\hat{\chi}_1\right]\hat{\hat{\zeta}}_2,\label{syseqs-j}\\
&\tilde{\tilde{\zeta}}_2(1+\chi_1\zeta_2)-\tilde{\tilde{z}}_2\left[\tilde{\zeta}_2+(x_1-\tilde{z}_2)\zeta_2\right]=\hat{\hat{\zeta}}_2(1+\chi_1\zeta_2)-\hat{\hat{z}}_2\left[\hat{\zeta}_2+(x_1-\hat{z}_2)\zeta_2\right].\label{syseqs-k}
\end{align}
\end{subequations}
Using equations \eqref{syseqs-a}, \eqref{syseqs-b}, \eqref{syseqs-c}, \eqref{syseqs-d}, \eqref{syseqs-e} and \eqref{syseqs-k} we express $\tilde{\tilde{x}}_1$, $\tilde{x}_1$, $\tilde{z}_2$, $\tilde{\chi}_1$, $\tilde{\zeta}_2$ and $\tilde{\tilde{\zeta}}_2$ in terms of ``$\tilde{\tilde{z}}_2-\hat{\hat{z}}_2$", as follows:
\begin{subequations}\label{termsw2}
\begin{align}
\tilde{\tilde{x}}_1&=\hat{\hat{x}}_1+\frac{x_1-\hat{z}_2+\chi_1\hat{\zeta}_2}{\hat{x}_1-z_2+\hat{\chi}_1\zeta_2+(\tilde{\tilde{z}}_2-\hat{\hat{z}}_2)(1-\chi_1\zeta_2)}(\tilde{\tilde{z}}_2-\hat{\hat{z}}_2),\\
\tilde{x}_1&=\hat{x}_1+(\tilde{\tilde{z}}_2-\hat{\hat{z}}_2),\\
\tilde{z}_2&=\frac{\hat{z}_2(\hat{x}_1-z_2+\hat{\chi}_1\zeta_2)+(\tilde{\tilde{z}}_2-\hat{\hat{z}}_2)\left[x_1+\chi_1(\hat{\zeta}_2-\hat{z}_2\zeta_2)\right]}{\hat{x}_1-z_2+\hat{\chi}_1\zeta_2+(\tilde{\tilde{z}}_2-\hat{\hat{z}}_2)(1-\chi_1\zeta_2)},\\
\tilde{\chi}_1&=\hat{\chi}_1-(\tilde{\tilde{z}}_2-\hat{\hat{z}}_2)\chi_1,\\
\tilde{\zeta}_2&=\hat{\zeta}_2+\frac{x_1-\hat{z}_2+\chi_1\hat{\zeta}_2}{\hat{x}_1-z_2+\hat{\chi}_1\zeta_2+(\tilde{\tilde{z}}_2-\hat{\hat{z}}_2)(1-\chi_1\zeta_2)}\zeta_2(\tilde{\tilde{z}}_2-\hat{\hat{z}}_2),\\
\tilde{\tilde{\zeta}}_2&=\hat{\hat{\zeta}}_2+\left[\hat{\zeta}_2(1-\chi_1\zeta_2)+(z_1-\hat{z}_2)\zeta_2\right](\tilde{\tilde{z}}_2-\hat{\hat{z}}_2).
\end{align}
\end{subequations}
Now, we substitute all the above to equation \eqref{syseqs-j} and, after a little manipulation, we can factorise \eqref{syseqs-j} in the following form:
\begin{equation}
(\tilde{\tilde{z}}_2-\hat{\hat{z}}_2)\left[b-c+(\hat{z}_2-x_1)(\hat{x}_1-z_2)(1-\zeta_2\chi_1)+(\hat{x}_1-z_2)\hat{\zeta}_2\chi_1+(x_1-\hat{z}_2)\zeta_2\hat{\chi}_1\right]=0,
\end{equation}
which implies $\tilde{\tilde{z}}_2=\hat{\hat{z}}_2$. 

Now, substituting $\tilde{\tilde{z}}_2=\hat{\hat{z}}_2$ to equations \eqref{termsw2}, we obtain $\tilde{\tilde{x}}_1=\hat{\hat{x}}_1$, $\tilde{x}_1=\hat{x}_1$, $\tilde{z}_2=\hat{z}_2$, $\tilde{\chi}_1=\hat{\chi}_1$, $\tilde{\zeta}_2=\hat{\zeta}_2$ and $\tilde{\tilde{\zeta}}_2=\hat{\hat{\zeta}}_2$.  Finally, equation \eqref{syseqs-i} implies $\tilde{\tilde{\chi}}_1=\hat{\hat{\chi}}_1$. Therefore, according to the trifactorisation property, map $\{\eqref{GmapYB}-\eqref{superYB}\}$ satisfies the parametric Yang-Baxter equation \eqref{SYB_eq1}.

Regarding the invariants, the supertrace of the quantity $\mathcal{L}_b(\textbf{y},\pmb{\psi})\mathcal{L}_a(\textbf{x},\pmb{\chi})$ reads:
\begin{equation*}
\str(\mathcal{L}_b(\textbf{y},\pmb{\psi})\mathcal{L}_a(\textbf{x},\pmb{\chi}))=2\lambda+(y_1-x_2)(x_1-y_2)-\chi_1\chi_2-\psi_1\psi_2-a-b-1.
\end{equation*}
Thus, $I=(y_1-x_2)(x_1-y_2)-\chi_1\chi_2-\psi_1\psi_2-a-b-1=I_1-I_2-a-b-1$ is invariant of the map $\{\eqref{GmapYB}-\eqref{superYB}\}$, where $I_1$ and $I_2$ are given by \eqref{sinvYB-a} and \eqref{sinvYB-b}, respectively. However, it can be readily verified that $I_1$ and $I_2$ are invariants themselves. In addition, it can be verified by straightforward calculation that $I_3$ and $I_4$ in \eqref{sinvYB-c} and \eqref{sinvYB-d}, respectively, are also invariants of the map $\{\eqref{GmapYB}-\eqref{superYB}\}$.

For the bosonic limit, we set $\chi_1=\chi_2=\psi_1=\psi_2=0$ in \eqref{superYB}, which implies $\xi_1=\xi_2=\eta_1=\eta_2=0$. Thus, the retrieved map coincides with map \eqref{YBliftKdV}.

\end{document}